\documentclass[pra,aps,twocolumn,10pt,showpacs,nofootinbib]{revtex4-2}
\usepackage[utf8]{inputenc}
\usepackage[T1]{fontenc}
\usepackage{amsmath,amsthm,charter}
\usepackage{amsfonts}
\usepackage{amssymb}
\usepackage{mathcmd}
\usepackage{graphicx} 

\usepackage{braket}
\usepackage{bm,bbold}
\usepackage{mathrsfs}
\usepackage{stmaryrd}
\usepackage{esint}
\usepackage{xcolor}

\usepackage{tikz}
\usepackage{rotating} 
\usepackage{float}
\usepackage{qcircuit}

\newtheorem{definition}{Definition}
\newtheorem{theorem}{Theorem}

\usepackage{ulem}
\usepackage[colorlinks=true, linkcolor=green, citecolor=blue]{hyperref}

\begin{document}

\title{Unbreakable and breakable quantum censorship}

\author{Julien Pinske}
    \email{julien.pinske@nbi.ku.dk}
    \affiliation{Paderborn University, Institute for Photonic Quantum Systems (PhoQS), Theoretical Quantum Science, Warburger Straße 100, 33098 Paderborn, Germany}

\author{Jan Sperling}
    \affiliation{Paderborn University, Institute for Photonic Quantum Systems (PhoQS), Theoretical Quantum Science, Warburger Straße 100, 33098 Paderborn, Germany}

\date{\today}

\begin{abstract}
    A protocol for regulating the distribution of quantum information between multiple parties is put forward.
    In order to prohibit the unrestricted distribution of quantum-resource states in a public quantum network, agents can apply a resource-destroying map to each sender's channel.
    Since resource-destroying maps only exist for affine quantum resource theories, censorship of a nonaffine resource theory is established on an operationally motivated subspace of free states. 
    This is achieved by using what we name a resource-censoring map.
    The protocol is applied to censoring coherence, reference frames, and entanglement.
    Because of the local nature of the censorship protocol, it is, in principle, possible for collaborating parties to bypass censorship.
    Thus, we additionally derive necessary and sufficient conditions under which the censorship protocol is unbreakable.
\end{abstract}
    
\maketitle

    \section{Introduction}

    As Shor's algorithm for the efficient factoring of prime numbers exemplified \cite{S94}, quantum information can be used to break certain cryptographic schemes \cite{MV18,SC21}, being foundational for quantum \cite{GR02,PA20} and post-quantum cryptography \cite{BL17}.
    Because of the prospects that modern information societies will one day be dealing with a quantum internet \cite{K08,SE18,IC22,RZ23}, in which quantum channels of increasing complexity connect numerous senders and receivers, establishing certain restrictions on the sharing of quantum resources becomes a subject of ever increasing interest.
    
    To prevent the unregulated spreading of quantum resources, such as coherence and entanglement, to malicious parties in their preparation of cryptographic attacks on critical infrastructures, governmental agencies might try to establish a form of \textit{quantum censorship}.
    In such a protocol, quantum information which is deemed benign crosses a network unaltered while hazardous quantum information is rendered classical (Fig. \ref{fig:motiv}).
    A less dystopian---but an information-processing equivalent---scenario might be the censorship of a commercialized network, with a provider offering free transmission of classical information, but demanding premium fees for sharing of quantum information.

    In this work, we devise a protocol for such quantum censorship applications.
    The protocol is based on a network of multiple sender-receiver pairs, being controlled by some dominant, protective agency (e.g., a governmental authority, a commercial provider, etc.) that applies a resource-destroying (RD) map \cite{LH17} locally to each sender. 
    This ensures that only free states of a quantum resource theory (QRT) \cite{CG19} are transmitted over the network.
    RD maps distinguish themselves from resource-breaking \cite{HS03,IK13}, resource-annihilating \cite{MZ10}, and resource-erasing protocols \cite{GP05,SB17} in that they destroy the quantum resource but do not alter free states.
    Moreover, RD maps are single-shot operations, thus avoiding costly procedures such as tomography by the agency.

\begin{figure}[b]
    \includegraphics[width=.3\textwidth]{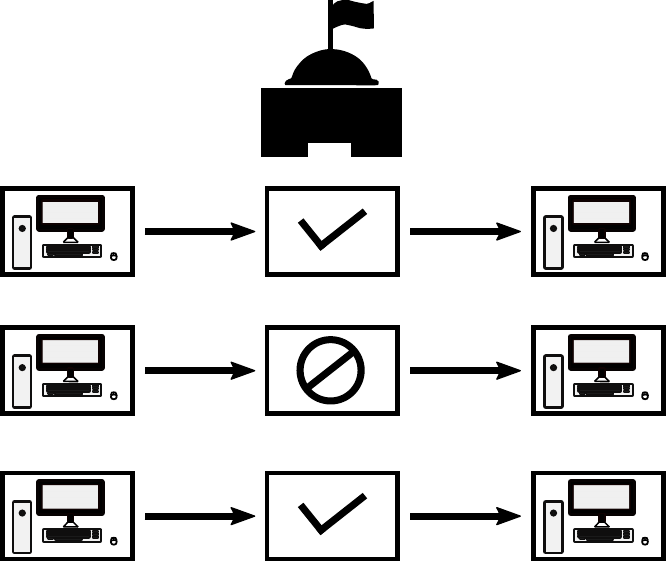}
    \caption{%
        In the quantum censorship protocol, a dominant, protective agency oversees quantum communication in a public-domain quantum network.
    }
    \label{fig:motiv}
\end{figure}
    
    The issue regarding RD maps is that they are not physically implementable for all QRTs.
    Indeed, only affine QRTs can give rise to a linear RD map;
    necessary and sufficient conditions for a QRT to have a (unique) RD map were derived in Ref. \cite{G17}.
    Examples of affine QRTs that possess RD maps include quantum coherence \cite{A06,BC14}, quantum thermodynamics \cite{BH15,SR20}, and quantum reference frames \cite{BR07,GS08}.
    On the other hand, there do not exist linear RD maps for real-valued quantum mechanics (affine) \cite{HG18,WK21}, quantum entanglement (convex) \cite{CD01,HH09,CV20}, quantum discord (nonconvex) \cite{OZ01}, and non-Gaussianity (convex \cite{TZ18} and nonconvex \cite{LR18}). 
    At first sight, this appears to set a fundamental limitation to which resources a quantum censorship can be imposed.
    In the case of a nonaffine QRT, we identify affine subspaces of free states on which a censorship can still be enforced.
    Since these subspaces are operationally motivated---while QRTs are physically motivated---, we introduce the notion of a resource-censoring (RC) map, being a generalization of RD maps.
    
    Once the censorship protocol is established, the question arises if the sending parties can use nonlocal resources, such as shared entanglement, to overcome the censorship, meaning that a resource reaches the receivers.
    This is, in principle, possible because the agent applies RC maps locally to each sender-receiver channel.
    Thus, we establish necessary and sufficient constraints on when collaborating parties can break the censorship.
    In particular, we find that a censorship that is realized via an RD map is unbreakable.
    It follows that the transmission of classical information (incoherent states) and speakable information (reference frames) can be enforced perfectly by the censorship protocol.
    By contrast, censorship of the (nonaffine) QRT of entanglement can be overcome by using preshared entanglement between multiple senders.
	Finally, the effect of noise on the protocol is discussed.

    \section{Quantum resource theories}

    When trying to establish censorship on quantum information, we first have to split the set of quantum states into resource states, whose distribution one wants to prevent, and free states, which propagate in the network unaltered.
    Making this distinction is the subject of QRTs \cite{CG19}.
    Each QRT comes with an assigned set of free states $\mathcal{F}(A)$, being a subset of the set of density operators, which we denote as $\mathcal{D}(A)$.
    The set $\mathcal{D}(A)$ contains positive semidefinite, unit-trace operators $\rho$ acting on the (here, finite-dimensional) Hilbert space $\mathcal{H}_A$ of a system $A$.
    
    A QRT is said to be affine, if the free states form an affine space; i.e., for any $\sigma_a\in \mathcal{F}(A)$, the state $\sigma=\sum_a t_a\sigma_a$, with $t_a\in\mathbb{R}$ and $\sum_a t_a=1$, is again a free state. 
    For $\mathcal{F}(A)\subseteq\mathcal{D}(A)$, its affine hull is here defined as 
    \begin{equation}
        \label{eq:affHull}
        \mathrm{Aff}(\mathcal{F}){=}\Big\{\sum_{a}t_a\sigma_a\,\Big|\,\sigma_a{\in}\mathcal{F}(A),\sum_a t_a{=}1\Big\}
        \cap\mathcal D(A).
    \end{equation}
    Since an affine combination of states $\sigma^a$ does not always yield a physical state, we made use of an intersection with the set of density operators $\mathcal{D}(A)$ in Eq. \eqref{eq:affHull} such that $\mathrm{Aff}(\mathcal{F})$ contains only physical states.%
    For example, this could be free states that admit a quasiprobability representation \cite{SW18} while resourceful states are non-decomposable;
    see Ref. \cite{PDBBSS21} for an experiment.

    Similarly, a QRT is said to be convex if the resource-free states form a convex set;
    i.e., for any $\sigma_a\in \mathcal{F}(A)$, the state $\sigma=\sum_a t_a\sigma_a$, with $t_a\geq 0$ and $\sum_a t_a=1$, is again a free state. 
    The convex hull of $\mathcal{F}(A)\subseteq\mathcal{D}(A)$ is
    \begin{equation}
        \label{eq:convHull}
        \mathrm{Conv}(\mathcal{F}){=}\Big\{\sum_{a}t_a\sigma_a\,\Big|\,\sigma_a{\in}\mathcal{F}(A),\sum_a t_a{=}1,t_a{\geq}  0\Big\}.
    \end{equation}
    Note that $\mathrm{Conv}[\mathcal{F}(A)]\subseteq \mathrm{Aff}[\mathcal{F}(A)]$ holds true because convex sums are a special case of affine sums where $0\leq t_a \leq 1$ is a probability.

    We denote by $\mathcal{D}(A_1\dots A_N)$ the set of quantum states of an $N$-partite composite system.
    The convex hull $\mathrm{Conv}[\mathcal{D}(A_1)\otimes\dots\otimes\mathcal{D}(A_N)]$ corresponds to the set of $N$-partite, fully separable states \cite{W89,H97}.
    Note that we here employ the following notation of tensor products of sets: $\mathcal D(A_1)\otimes\dots\otimes\mathcal D(A_N)=\{\rho_{A_1}\otimes\dots\otimes\rho_{A_n}|\rho_{A_1}\in\mathcal D(A_1),\ldots,\rho_{A_N}\in\mathcal D(A_N)\}$.
    Further, we suppose that the set of composite free states $\mathcal{F}(A_1\dots A_N)$ contains at least $\mathcal{F}(A_1)\otimes\dots \otimes\mathcal{F}(A_N)$ \cite{CG19}.
    This means that the independent preparation of free states by multiple parties gives a free state on the composite system.
    Moreover, discarding subsystems may not create a resource;
    i.e., for $\sigma\in\mathcal{F}(A_1\dots A_N)$, its marginals $\mathrm{Tr}_{a}(\sigma)\in \mathcal{F}(A_1\dots A_{a-1}A_{a+1}\dots A_N)$, for $a=1,\dots,N$, are free, too.
    Therefore, if $\mathcal{F}(A_1),\dots,\mathcal{F}(A_N)$ are affine, then one defines
    \begin{equation}
        \label{eq:AffFree}
        \mathcal{F}(A_1\dots A_N)
        =
        \mathrm{Aff}[\mathcal{F}(A_1)\otimes\dots\otimes\mathcal{F}(A_N)].
    \end{equation}
    If $\mathcal{F}(A_1),\dots,\mathcal{F}(A_N)$ are convex, then one has
    \begin{equation}
        \label{eq:ConvFree}
        \mathcal{F}(A_1\dots A_N)=\mathrm{Conv}(\mathcal{F}(A_1)\otimes\dots\otimes\mathcal{F}(A_N)).
    \end{equation}
    And, for a general QRT, $\mathcal{F}(A_1\dots A_N)\supseteq\mathcal{F}(A_1)\otimes\dots\otimes\mathcal{F}(A_N)$ holds true.

    \subsection{Resource-destroying maps and resource-censoring maps}

    Physical operations are mathematically expressed as quantum channels \cite{W18}, i.e., linear maps $\Lambda:\mathcal{D}(A)\to\mathcal{D}(B)$ that are completely positive and trace-preserving.
    A quantum channel $\Delta$ is said to be RD, if it additionally satisfies
    \begin{align}
        \text{(i) }
        \forall\rho\in\mathcal D(A):& \quad\Delta(\rho)\in\mathcal F(B),
        \tag{resource-destroying}
        \\
        \label{eq:NRF}
        \text{(ii) }
        \forall\sigma\in\mathcal F(A):& \quad\Delta(\sigma)=\sigma.
        \tag{freeness-preserving}
    \end{align}

    In Ref. \cite{LH17}, it was shown that the existence of a RD map implies that $\mathcal{F}(A)$ is affine; see also Refs. \cite{G17,CG19}.
    To see this, let $\sigma=\sum_a t_a\sigma_a\notin\mathcal{F}(A)$ be an affine combination, and $\sigma_a\in\mathcal{F}(A)$ are free states. 
    If an RD map $\Delta$ would exist for a nonaffine theory, then $\Delta\big(\sum_a t_a\sigma_a\big)=\sum_a t_a\Delta(\sigma_a)$ has to be a free state by condition (i).
    On the other hand, condition (ii) implies $\Delta(\sigma_a)=\sigma_a$ for all $a$, thus $\Delta(\sigma)=\sigma$.
    However, $\sigma$ is not a free state by the initial assumption.
    Thus, we are left with a contradiction, showing that such a map $\Delta$ cannot exist.
    
    Nevertheless, for a nonaffine QRT, a generalization of RD maps can be introduced, which we dub RC maps.
    \begin{definition}
        A channel $\Delta^\prime$ is said to be RC if it satisfies
        \begin{align}
            \forall\rho\in\mathcal D(A):& ~\Delta^\prime(\rho)\in\mathcal F(B),
            \tag{resource destroying}
            \\
            \label{eq:NRF}
            \forall\sigma\in\mathcal F^\prime(A):&  ~\Delta^\prime(\sigma)=\sigma,
            \tag{almost freeness-preserving}
        \end{align}
    where $\mathcal{F}^\prime(A)\subseteq\mathcal{F}(A)$ is a chosen affine subspace.
    \end{definition}
    
    We emphasize that RC maps are not just RD maps belonging to a smaller QRT $\mathcal{F}^\prime(A)$, with $\mathcal{F}^\prime(A)\subseteq\mathcal{F}(A)$, since we do not demand an RC map $\Delta^\prime$ to map any state $\rho\in\mathcal{D}(A)$ onto a state in $\mathcal{F}^\prime(B)$.
    Moreover, free states $\mathcal{F}(A)$ of a QRT are commonly motivated by physical limitations while the set $\mathcal{F}^\prime(A)$ is motivated operationally.
    Simply speaking, $\mathcal{F}^\prime(A)$ is a subset of free states for which one can guarantee that these pass the channel unaltered while any resource $\rho\notin\mathcal{F}(A)$ is destroyed. 
    The free states in $\mathcal{F}(A)\setminus\mathcal{F}^\prime(A)$ might also undergo changes.
    Since $\mathcal{F}^\prime(A_1),\ldots,\mathcal{F}^\prime(A_N)$ are, by definition, affine, $\mathcal{F}^\prime(A_1\dots A_N)$ is given by the affine hull in Eq. \eqref{eq:AffFree}.
    For an affine QRT, one has $\Delta^\prime=\Delta$, with $\mathcal{F}^\prime(A)=\mathcal{F}(A)$.
    On the other hand, for any nonaffine QRT $\mathcal{F}(A)$, one can always find at least a minimal construction $\mathcal{F}^\prime(A)=\{\sigma\}$, containing a single state $\sigma\in\mathcal{F}(A)$;
    then, $\Delta^\prime(\rho)=\mathrm{Tr}(\rho)\sigma$ is RC.

    \section{Quantum censorship}

    In the following, the quantum censorship protocol is introduced.
    Firstly, the protocol is studied for noiseless channels. 
    Secondly, we discuss under which circumstances multiple senders can coordinate their resources to potentially overcome censorship.
    After the discussion of important special cases, the effect of noise on the protocol is investigated. 

    \subsection{Censorship over noiseless channels}

    Consider $N$ senders $A_1,\dots,A_N$ who have access to local quantum resources, e.g., party $A_a$ can prepare any state $\rho_{A_a}\in \mathcal{D}(A_a)$. 
    In an unregulated network, each sender is connected to one of the receivers $B_1,\dots,B_N$ via the noiseless channel $\mathrm{id}_{A_a\to B_a}$.
    However, in order to prevent the transmission of resource states, an agent sits in between each sender-receiver pair.
    The agent's goal is to limit the type of quantum states that can be shared between parties to the free states $\mathcal{F}(A_a)$ of a QRT.
    The agent informs the senders that only the transmission of free states in an affine subspace $\mathcal{F}^\prime(A_a)\subseteq\mathcal{F}(A_a)$ is authorized, the user agreement.
    To enforce that policy, the agent can implement an RC map $\Delta^\prime$.
    Thus, the information processing protocol of (noiseless) quantum censorship is 
    \begin{equation*}
    	\Qcircuit @C=1em @R=.7em {
    		\lstick{A_1} & \qw &\gate{\Delta^{\prime}} & \qw & \qw & B_1\\
    		\lstick{\vdots} & &\vdots & & & \vdots\\
    		\lstick{} & & & & & \\
    		\lstick{A_N} & \qw &\gate{\Delta^{\prime}} & \qw & \qw & B_N.\\
    	}
    \end{equation*}
    
    As long as each sender $A_1,\dots,A_N$ only has access to local quantum resources, i.e., the composite system is in a product state $\rho_{A_1}\otimes\dots\otimes\rho_{A_N}$, receiving parties $B_1,\dots,B_N$ obtain $\Delta^\prime(\rho_{A_1})\otimes\dots\otimes \Delta^\prime(\rho_{A_N})$, which is a free state in $\mathcal{F}(B_1\dots B_N)$, as intended by the agent.
    If, however, an initial state $\sigma_{A_1}\otimes\dots\otimes \sigma_{A_N}$ belongs to $\mathcal{F}^\prime(A_1)\otimes \dots\otimes \mathcal{F}^\prime(A_N)$, it remains unchanged by the action of $(\Delta^\prime)^{\otimes N}$.
    This allows the users of the network to carry out (undisturbed) communication only with messages $\sigma\in\mathcal{F}^\prime(A_1\ldots A_N)$.
    
    \subsection{Breakable and unbreakable censorship}

    Clearly, a single sender cannot break censorship as $\rho\in\mathcal{D}(A)$ is mapped onto a free state $\Delta^\prime(\rho)\in\mathcal{F}(B)$.
    But sending parties $A_1,\dots,A_N$ might coordinate their actions to prepare a nonlocal resource state $\rho\in\mathcal{D}(A_1\dots A_N)$.
    In this case, the circuit is
	\begin{equation*}
		\Qcircuit @C=1em @R=.7em {
			\lstick{A_1} & \ctrl{1} &\gate{\Delta^\prime} & \qw & \qw & B_1\\
			\lstick{\vdots} &  &\vdots & & & \vdots\\
			\lstick{} &  & & & & \\
			\lstick{A_N} & \ctrl{-1} &\gate{\Delta^{\prime}} & \qw & \qw & B_N,\\
		}
	\end{equation*}
	where the vertical line between the senders indicates preshared entanglement (and randomness \cite{CG19}) necessary to prepare an arbitrary $N$-partite quantum state.
    Because of the local action of the agent's operation $\Delta^{\prime}$, the question arises if the censorship can be overcome in this manner?
    Formally, we define the notion of breakable censorship as follows.
    \begin{definition}
        A censorship is breakable if there exists a state $\rho\notin\mathcal{F}(A_1\dots A_N)$ such that $(\Delta^{\prime})^{\otimes N}(\rho)=\rho$ .
        Otherwise, censorship is said to be unbreakable.
    \end{definition}
    Simply speaking, censorship is breakable if a quantum correlated resource state reaches the receivers unaltered.
    The receivers can coordinate their actions to make use of the resource.
    When censorship is unbreakable, malicious users $A_1,\dots,A_N$ cannot proliferate quantum resources, thus making it easier to attribute the origin of a cryptographic attack in places, wherever post-quantum cryptography is not at its state-of-the-art.
    In a commercial setting, where a provider demands premium fees for sharing quantum resources, overcoming the censorship creates a free-rider problem, in which users can transmit quantum information without paying. 
    This, in turn, destroys a provider's incentive to participate in the build-up of a global quantum internet.
    
    The following theorem establishes for which QRTs the censorship can be overcome.
    \begin{theorem}
        \label{th:censor}
        A censorship is breakable, if and only if $\mathcal{F}^\prime(A_1\dots A_N)\setminus \mathcal{F}(A_1\dots A_N)$ is nonempty.
    \end{theorem}
    \begin{proof}
        If $\mathcal{F}^\prime(A_1\dots A_N)\setminus \mathcal{F}(A_1\dots A_N)$ is nonempty,
        then there exists a resource state $\rho\notin\mathcal{F}(A_1\dots A_N)$ that is stabilized by $(\Delta^\prime)^{\otimes N}$, i.e., $(\Delta^\prime)^{\otimes N}(\rho)=\rho$.
        This follows from the linearity of $\Delta^\prime$ and the definition of the affine hull $\mathcal{F}^\prime(A_1\dots A_N)$;
        see Eq. \eqref{eq:AffFree}.
        Hence, $(\Delta^{\prime})^{\otimes N}(\rho)=\rho\notin\mathcal{F}(B_1\dots B_N)$ and censorship is breakable.
        Conversely, if censorship is breakable, then there exists a state $\rho\notin\mathcal{F}(A_1\dots A_N)$ such that $(\Delta^\prime)^{\otimes N}(\rho)=\rho$. By the above argument, $\rho$ lies in $\mathcal{F}^\prime(A_1\dots A_N)\setminus \mathcal{F}(A_1\dots A_N)$, completing the proof.
    \end{proof}

    Intuitively, this breaking of the censorship can be understood as follows. 
    The subset $\mathcal{F}^\prime(A)$ of the free states $\mathcal{F}(A)$ was motivated operationally as a space on which one could establish a censorship for a single sender-receiver pair. 
    However, its affine hull $\mathcal{F}^\prime(A_1\dots A_N)$ as defined in Eq. \eqref{eq:AffFree} might contain states that are resourceful on the composite system.

    \subsection{Special cases}

    For the case $\Delta^\prime=\Delta$, being the (unique \cite{G17}) RD map of a QRT, we have the following theorem. 
    \begin{theorem}
        \label{co:RD}
        Let $\Delta$ be the RD map of a QRT.
        Then, the censorship is unbreakable.
    \end{theorem}
    \begin{proof}
        Since $\Delta$ is RD, the set of free states $\mathcal{F}(A)$ is affine.
        Hence, $\mathcal{F}(A_1\dots A_N)=\mathcal{F}^\prime(A_1\dots A_N)$, and by virtue of Theorem \ref{th:censor}, the censorship is unbreakable. 
    \end{proof}

    In principle, the censorship protocol can be made unbreakable for any QRT.
    This can be achieved by choosing $\mathcal{F}^\prime(A)=\{\sigma\}$ as a single-state edge case.
    The RC map of the theory is the replacement channel $\Delta^\prime(\rho)=\mathrm{Tr}(\rho)\sigma$.
    Then, $\mathcal{F}^\prime(A_1\dots A_N)=\mathcal{F}^\prime(A_1)\otimes\dots\otimes\mathcal{F}^\prime(A_N)$, which is always contained in $\mathcal{F}(A_1\dots A_N)$. 
    It follows from Theorem \ref{th:censor} that censorship is unbreakable.

    Next, suppose we are concerned with the censorship of a convex QRT.
    The set $\mathcal{F}(A_1\dots A_N)$ as defined by the convex hull in Eq. \eqref{eq:ConvFree} contains free, $N$-partite separable states.
    Thus, entanglement-breaking channels play a distinct role in the censorship of these resources.
    The channel $\Delta^\prime$ is entanglement breaking \cite{HS03} if $\mathrm{id}_{A_1\to B_1}\otimes\Delta^\prime$ maps any bipartite state $\rho$ onto a separable state; i.e., $(\mathrm{id}_{A_1\to B_1}\otimes\Delta^\prime)(\rho)$ is an element of $\mathrm{Conv}(\mathcal{D}(B_1)\otimes\mathcal{D}(B_2))$. 
    As $\Delta^\prime$ is additionally RC, $\mathrm{id}_{A_1\to B_1}\otimes\Delta^\prime$ is a projection onto $\mathrm{Conv}(\mathcal{D}(B_1)\otimes\mathcal{F}^\prime(B_2))$.
    For an entanglement-breaking RC map, the following theorem holds true.

    \begin{theorem}
        \label{th:EB}
        If $\Delta^\prime$ is an entanglement-breaking RC map of a convex QRT,
        then censorship is unbreakable.
    \end{theorem}
    \begin{proof}
        Let $\Delta^\prime$ be entanglement breaking. 
        Then, $(\Delta^\prime)^{\otimes N}$ is a mapping from $\mathcal{D}(A_1\dots A_N)$ to the set $\mathrm{Conv}(\mathcal{F}^\prime(B_1)\otimes\dots\otimes \mathcal{F}^\prime(B_N))$.
        But since $(\Delta^\prime)^{\otimes N}$ stabilizes states in the affine hull $\mathcal{F}^\prime(A_1\dots A_N)$ given by Eq. \eqref{eq:AffFree}, it follows that 
        \begin{equation}
        \begin{aligned}
            \mathcal{F}^\prime(A_1\dots A_N)
            &\subseteq \mathrm{Conv}(\mathcal{F}^\prime(A_1)\otimes\dots\otimes \mathcal{F}^\prime(A_N))
            \\
            &\subseteq\mathcal{F}(A_1\dots A_N),
        \end{aligned}
        \end{equation}
        where the second line follows from $\mathcal{F}^\prime(A)\subseteq\mathcal{F}(A)$, and because $\mathcal{F}(A_1\dots A_N)$ is defined via the convex hull in Eq. \eqref{eq:ConvFree}.
        Hence, $\mathcal{F}^\prime(A_1\dots A_N)\setminus \mathcal{F}(A_1\dots A_N)$ is empty and Theorem \ref{th:censor} implies unbreakable censorship.
    \end{proof}

    For the purpose of illustration, we can make use of the minimal construction $\mathcal{F}^\prime(A)=\{\sigma\}$ in which $\sigma$ is a single free state belonging to a convex QRT $\mathcal{F}(A)$.
    An RC map is given by the replacement channel $\Delta^\prime(\rho)=\mathrm{Tr}(\rho)\sigma$.
    Since $\Delta^\prime$ is entanglement breaking, Theorem \ref{th:EB} ensures that the censorship is unbreakable, which might also be obvious from the form of the channel $\Delta^\prime$.
    
    \subsection{Censorship via noisy channels}

    So far, we have restricted ourselves to perfect communication;
    that is, each sender is connected to a receiver via an identity channel.
    In realistic communication scenarios, however, we expect information transmission to be performed over a noisy channel $\Phi:\mathcal{D}(A_a)\to\mathcal{D}(A_a)$.
    While this is a well-known issue for any information processing task, in the context of the censorship protocol, we ought to be worried that the RC map $\Delta^\prime$ introduces additional errors.
    Thus, we consider the noisy process $\Phi$ to occur before the operation $\Delta^\prime$.
    The protocol for the noisy case reads
    \begin{equation*}
    	\Qcircuit @C=1em @R=.7em {
    		\lstick{A_1} & \gate{\Phi} &\gate{\Delta^\prime} & \qw & \qw & B_1\\
    		\lstick{\vdots} &\vdots &\vdots & & & \vdots\\
    		\lstick{} & & & & & \\
    		\lstick{A_N} & \gate{\Phi} &\gate{\Delta^\prime} & \qw & \qw & B_N.\\
    	}
    \end{equation*}

	Throughout, it is assumed that $\Phi$ is resource non-generating \cite{CG19,LH17};
    that is, for any free state $\sigma\in\mathcal{F}(A_a)$, one has $\Phi(\sigma)\in\mathcal{F}(A_a)$.
	This seems to be a reasonable assumption because we rarely expect a noisy map $\Phi$ to create a resource from a free state.
    If censorship is established via a RD map $\Delta^\prime=\Delta$, then $\Delta(\sigma)=\sigma$ for any $\sigma\in\mathcal{F}(A)$.
	This implies that the noise $\Phi$ commutes with $\Delta$ on the set of free states, i.e.,
    \begin{equation}
        \forall \sigma\in\mathcal{F}(A): (\Delta\circ\Phi)(\sigma)=(\Phi\circ\Delta)(\sigma).
    \end{equation}
	In this scenario, the censorship protocol does not introduce additional errors through the RD map $\Delta$.
    This means that if a sender transmits only free states---as the agent wants them to do---, their message $\sigma$ is obtained by the receiver as $\Phi(\sigma)$.
    Of course, noiseless communication is infeasible in real-world settings, but the agent (e.g., a network provider) can aim at high-fidelity communication, avoiding the introduction of additional noise by enforcing censorship via $\Delta$.
 
    The situation is more delicate if we consider censorship using a RC map $\Delta^\prime$.
    The RC map stabilizes only an affine subspace $\mathcal{F}^\prime(A)\subseteq\mathcal{F}(A)$.
    States in $\mathcal{F}(A)\setminus \mathcal{F}^\prime(A)$ could be altered by $\Delta^\prime$. 
    Then, any resource-non-generating $\Phi$ that takes elements in $\mathcal{F}^\prime(A)$ to elements in $\mathcal{F}(A)\setminus \mathcal{F}^\prime(A)$ does not generate a resource, but it might lead $\Delta^\prime$ to alter these states.
    The protocol then introduces additional changes to the state that distort the sender-receiver experience in addition to the already present noise generated by $\Phi$.
    Thus, in order to ensure that customers can exchange free states in $\mathcal{F}^\prime(A)$ without interference caused by the RC map $\Delta^\prime$, the provider must, in general, keep the sender-receiver channels free of any noise process that is not an automorphism $\Phi^\prime:\mathcal{F}^\prime(A)\to \mathcal{F}^\prime(A)$.

    On the other hand, due to $\Delta^\prime$ being a projection onto $\mathcal{F}^\prime$, there might be practical situations in which the action of $\Delta^\prime$ has a correcting effect, i.e., $\Delta^\prime(\Phi(\sigma))$ is closer (with respect to some metric) to $\sigma$ that the noisy message $\Phi(\sigma)$.

    \section{Censorship of specific resources}

    In the following, the censorship protocol is illustrated for several resources including coherence, reference frames, and entanglement.

    \subsection{Censorship of coherence}

    In the QRT of coherence \cite{BC14,SV15}, one quantifies the amount of superpositions in a general mixed state with respect to a fixed orthonormal basis $\{\ket{x}\}_x$, the incoherent basis. 
    Free (likewise, incoherent) states admit a diagonal representation in that basis, $\sigma=\sum_x p_x\ket{x}\bra{x}$.
    This QRT is affine since, by Eqs. \eqref{eq:affHull} and \eqref{eq:convHull}, the definitions of convex and affine hull coincide in the example under study.
    An RD map is given by the completely dephasing channel \cite{LH17,CG19}
    \begin{equation}
        \label{eq:RDcoh}
        \Delta(\rho)=\sum_{x}\ket{x}\bra{x}\rho \ket{x}\bra{x}.
    \end{equation}
    Imposing censorship on coherence using $\Delta$ means that only incoherent states (classical information) are preserved during the communication. 
    Since $\Delta$ is RD, the censorship is unbreakable, Theorem \ref{co:RD}.
    This is a positive result for any provider (agent) trying to reserve quantum communication for specific costumers, while restricting the general users of the network to classical communication only.
    Senders have to accept such policies as there is no way of breaking the censorship.

    A physical realization of the censorship can be implemented by linear optics. 
    The sender $A$ prepares the coherent superposition $\ket{\psi}=\alpha_H\ket{H}+\alpha_V\ket{V}$ in their lab.
    Here, horizontal and vertical polarization $\ket{H}$ and $\ket{V}$ define the incoherent basis, with $|\alpha_H|^2+|\alpha_V|^2=1$.
    To prevent the transmission of coherent quantum information to $B$, the agent simply applies a polarization filter to the state $\ket{\psi}$.
    This realizes a projective measurement of $\ket{H}\bra{H}$ or $\ket{V}\bra{V}$, depending on the filter.
    Since the agent conceals which measurement was performed, $B$'s best description is given by the incoherent state $\sigma=|\alpha_H|^2\ket{H}\bra{H}+|\alpha_V|^2\ket{V}\bra{V}$.

    \subsection{Censorship of reference frames}

    Certain types of quantum information are, without a common reference frame, of no use to the communicating parties in a network.
    For instance, in the QRT of coherence, one can only decide if a given state is free or not if the incoherent basis $\{\ket{x}\}_x$ is known.
    Mathematically, we describe a change of the reference frame by a unitary operator $U_a$, relating a sender's state $\rho$ to a receiver's state via $U_a\rho U_a^\dagger$.
    However, if $U_a$ is unknown, the description of the state is obtained by averaging over all possible values in a group $\mathcal{G}=\{U_a\}_a$, i.e., $\Delta(\rho)=\tfrac{1}{\vert \mathcal{G}\vert}\sum_{a=1}^{\vert \mathcal{G}\vert} U_{a}\rho U_{a}^\dagger$.
    The channel $\Delta$ is also known as the $\mathcal{G}$-twirling map \cite{CG19}.
    If we consider a lack of a shared reference frame to define the free states $\mathcal{F}(A)=\big\{\Delta(\rho)\,\big|\,\rho\in\mathcal{D}(A)\big\}$, then $\Delta$ is a projection onto $\mathcal{F}(A)$.
    In particular, due to $\mathcal{F}(A)$ being affine, it is the RD map of the QRT.
    Thus, the censorship of reference frames in unbreakable; see Theorem \ref{co:RD}.
    Note that the same conclusion cannot be reached using Theorem \ref{th:EB}. 
    Even though $\mathcal{F}(A)$ is affine, and thus convex, $\Delta$ is generally not entanglement breaking.

    \subsection{Censorship of entanglement}

    In the QRT of entanglement \cite{HH09,CG19}, the set of free states $\mathcal{F}(A)$ contains all separable bipartite states of the system $A$, i.e., $\mathcal{H}_A=\mathcal{X}_{A}\otimes \mathcal{Y}_{A}$. 
    A quantum state $\sigma\in\mathcal{D}(A)$ is said to be separable if it can be written as a probabilistic mixture of pure product states \cite{W89} 
    \begin{equation}
        \sigma=\sum_x p_x \ket{\psi^x}\bra{\psi^x}_{\mathcal{X}_A}\otimes \ket{\phi^x}\bra{\phi^x}_{\mathcal{Y}_A},
    \end{equation}
    where $p_x\geq 0$ and $\sum_xp_x=1$.
    Mathematically, $\mathcal{F}(A)$ is the convex hull of $\mathcal{D}(\mathcal{X}_A)\otimes \mathcal{D}(\mathcal{Y}_A)$.
    Since $\mathcal{F}(A)$ is convex, but not affine, there is no RD map for the theory.
    However, if the agent informs a sender-receiver pair to agree on a fixed orthonormal basis $\{\ket{x}\}_x$ for their subsystems $\mathcal{X}_A$ and $\mathcal X_B$, a censorship between them can be established.
    This is done on the affine subspace $\mathcal{F}^\prime(A)$, containing classical-quantum states $\sigma=\sum_{x}p_x\ket{x}\bra{x}_{\mathcal{X}_A}\otimes \sigma^x_{\mathcal{Y}_A}$ \cite{D10}, which are diagonal with respect to $\{\ket{x}\}_x$ in $\mathcal{X}_A$, and we have arbitrary $\sigma_{\mathcal{Y}_A}^x$.
    To see that $\mathcal{F}^\prime(A)$ is indeed affine, consider the affine combination $\sigma=\sum_a t_a\sigma^a$ of free states $\sigma^a=\sum_{x}p_{x,a}\ket{x}\bra{x}_{\mathcal{X}_A}\otimes \sigma_{\mathcal{Y}_A}^{x,a}$, viz.
    \begin{equation}
            \sigma=\sum_{x} q_x\ket{x}\bra{x}_{\mathcal{X}_A} \otimes \omega^x_{\mathcal{Y}_A}\in\mathcal{F}^\prime(A),
    \end{equation}
    where $\omega^x_{\mathcal{Y}_A}{=}\sum_a t_a p_{x,a}\sigma^{x,a}_{\mathcal{Y}_A}/\sum_a t_a p_{x,a}$ and $\sum_{a,x} t_a p_{x,a} =\sum_{x}q_x=1$. 
    An RC map for entanglement is defined as
    \begin{equation}
        \label{eq:RDent}
        \Delta^\prime(\rho)=\sum_x\big(\ket{x}\bra{x}_{\mathcal{X}_{A}}\otimes\mathbb{1}_{\mathcal{Y}_{A}}\big)\rho \big(\ket{x}\bra{x}_{\mathcal{X}_{A}}\otimes\mathbb{1}_{\mathcal{Y}_{A}}\big).
    \end{equation}
    Note that $\Delta'$ acts trivially on the subsystems $\mathcal{Y}_A$, i.e., $\Delta^\prime=\Delta\otimes\mathrm{id}_{\mathcal{Y}_A\to\mathcal Y_B}$, with $\Delta$ being the dephasing channel, Eq. \eqref{eq:RDcoh}.
    The channel in Eq. \eqref{eq:RDent} can be used to impose censorship on entanglement.
    To see this, note that $\Delta^\prime$ takes any bipartite state $\rho\in\mathcal{D}(A)$ to a separable state in $\mathcal{F}(B)$. 
    However, not all free states $\sigma\in\mathcal{F}(A)$ are stabilized by the map in Eq. \eqref{eq:RDent}, but only those in $\mathcal{F}^\prime(A)$.

    In the QRT of entanglement, $\mathcal{F}(A_1\dots A_N)$ is given by the convex hull in Eq. \eqref{eq:ConvFree}.
    In contrast, the affine hull $\mathcal{F}^\prime(A_1\dots A_N)$, as defined in Eq. \eqref{eq:AffFree}, contains states possessing entanglement between senders $A_1,\dots, A_N$.
    It follows that $\mathcal{F}^\prime(A_1\dots A_N)\setminus \mathcal{F}(A_1\dots A_N)$ is nonempty, and Theorem \ref{th:censor} thus implies that censorship can be broken.
    To see this explicitly, consider two senders $A_1$ and $A_2$ sharing the state
    \begin{equation}
        \rho_{A_1A_2}=\sum_{x,y}p_{xy}\ket{x}\bra{x}_{\mathcal{X}_{A_1}}\otimes \ket{y}\bra{y}_{\mathcal{X}_{A_2}}\otimes \rho^{xy}_{\mathcal{Y}_{A_1}\mathcal{Y}_{A_2}},
    \end{equation}
    where at least one $\rho^{xy}$ is entangled.
    The protocol for this case is
    \begin{equation*}
    	\Qcircuit @C=1em @R=.7em {
    		\lstick{} & \qw &\gate{\Delta} & \qw & \qw\\
    		\lstick{}& \ctrl{2} &\qw & \qw & \qw 
    		\inputgroup{1}{2}{.85em}{A_1}\\
    		\lstick{} & \qw & \gate{\Delta} & \qw & \qw \\
    		\lstick{} & \ctrl{-2} &\qw & \qw & \qw & . \inputgroup{3}{4}{0.9em}{A_2}\\
    	}
    \end{equation*}
    It is not hard to see that $\Delta^\prime_{A_1} \otimes \Delta^\prime_{A_2}$ leaves the states $\rho^{xy}$ unaltered.
    Thus, entanglement is passed on to the receivers, and censorship has been broken.
    Note that the agent might enforce an unbreakable censorship of entanglement by resorting to a (stricter) censorship on coherence, using the RD map $\Delta^{\otimes 2}$ from Eq. \eqref{eq:RDcoh}, and redefining $\mathcal{F}^\prime(A)$ as the set of (bipartite) incoherent states.
    This suffices because there cannot be entanglement without coherence, and censorship of coherence is unbreakable. 

	In general, the RC map in Eq. \eqref{eq:RDent} does not commute with resource non-generating (i.e., non-entangling \cite{HN03,VH05}) noise $\Phi$ on the set of free states $\mathcal{F}(A)$.
	Thus, the agent must be worried that their action $\Delta^\prime$ introduces additional errors into the message, despite a sender $A$ using the network permissibly, i.e., sending only messages $\sigma\in\mathcal{F}^\prime(A)$ according to the user agreement. 
 
	To illustrate the undesirable effects such noise can have on a state $\sigma=\sum_{x}p_x\ket{x}\bra{x}_{\mathcal{X}_A}\otimes \sigma^x_{\mathcal{Y}_A}$, consider a swap channel $\Phi(\rho\otimes\sigma)=\sigma\otimes\rho$. 
	Clearly, the noise $\Phi$ is resource non-generating (non-entangling).
	After the agent, applies the RC map in Eq. \eqref{eq:RDent}, a receiver $B$ is left with the incoherent state
	\begin{equation}
		(\Delta^\prime\circ\Phi)(\sigma)
        {=}\sum_{x,y}p_x\bra{y}\sigma^x_{\mathcal Y_A}\ket{y} \ket{y}\bra{y}_{\mathcal{X}_B}\otimes\ket{x}\bra{x}_{\mathcal{Y}_B}.
	\end{equation}
	This leaves $B$ without any quantum properties left in the state.
	While a swap channel might be a non-intuitive form of noise, mixing up a bit sequences is certainly possible, and it highlights some of the difficulties (quantum) network providers (i.e., the agent) faces when trying to establish a quantum censorship, while still keeping the network operable.

    A physical realization of the censorship can be devised using linear optics. 
    Two senders $A_1$ and $A_2$ prepare the entangled state
    \begin{equation*}
        \rho_{A_1 A_2}=\ket{H}\bra{H}_{\mathcal{X}_{A_{1}}}\otimes\ket{V}\bra{V}_{\mathcal{X}_{A_{2}}}\otimes\ket{\phi^+}\bra{\phi^+}_{\mathcal{Y}_{A_{1}}\mathcal{Y}_{A_{2}}}.
    \end{equation*}
    Here, horizontal and vertical polarization $\ket{H}$ and $\ket{V}$ define the incoherent basis and $\ket{\phi^+}=(\ket{HH}+\ket{VV})/\sqrt{2}$ is a Bell state.
    The agent tries to prevent the transmission of entanglement to receivers $B_1$ and $B_2$ using the RC map \eqref{eq:RDent}. 
    There, the agent performs a (nonselective) polarization measurement on the subsystems $\mathcal{X}_{A_{1}}$ and $\mathcal{X}_{A_{2}}$, thus realizing a dephasing with respect to the states $\ket{H}$ and $\ket{V}$ [see Eq. \eqref{eq:RDcoh}].
    See also Ref. \cite{UAC18} for experimental results on a controllable dephasing channel. 
    On the other hand, the agent does not apply any operation to the subsystems $\mathcal{Y}_{A_{1}}$ and $\mathcal{Y}_{A_{2}}$ [see Eq. \eqref{eq:RDent}].
    It follows that the state $\rho_{A_1 A_2}$ is not altered by the RC map $\Delta^\prime_{A_1}\otimes\Delta^\prime_{A_2}$. 
    The entangled state $\ket{\phi^+}\bra{\phi^+}$, thus reaches the receivers $B_1$ and $B_2$.
    The censorship has been broken.

    \section{Conclusion}

    We introduced a protocol for quantum censorship. 
    Therein, an agent can apply RC or RD maps locally to each sender-receiver connection, thus prohibiting the distribution of resource states through the network at will.
    By using such maps, the protocol avoids any measurements of a state, which would render the network unusable for quantum communication.
    Since RD maps exist only for affine QRTs, RC maps were utilized to impose censorship on an affine subspace of free states.
    Our necessary and sufficient conditions reveal under which censorship is unbreakable.
    This was the case for the QRT of coherence and reference frames while the censorship of entanglement could be overcome.

    Quantum censorship protocol becomes especially urgent once we are confronted with the emergence of a widely accessible quantum internet.
    See, for instance, Refs. \cite{NB22,ZS22} for recent experimental progress in this direction.
    On the one hand, quantum censorship allows governmental authorities to prevent ill-intentioned parties from quantum-cryptographic attacks.
    On the other hand, commercial enterprises may offer free classical services but want to charge premium fees for quantum communication.
    Also, future studies of more advanced (possibly non-local) censorship protocols might be a worthwhile endeavor.
    We hope our work paves the way for a discussion of quantum censorship as a previously unappreciated tool in quantum communication.

    \acknowledgments
    This work received financial support through Ministry of Culture and Science of the State of North Rhine-Westphalia (Project PhoQC).

\end{document}